\documentclass[a4paper,thm-restate]{lipics-v2021}
\usepackage[utf8]{inputenc}

\usepackage{amsmath,amsfonts,amssymb,xspace}
\usepackage{ mathrsfs }
\usepackage{amsthm}
\usepackage{graphicx}
\usepackage{paralist}
\usepackage{mathtools}
\usepackage{xcolor}

\usepackage{algorithm}
\usepackage[noend]{algpseudocode}
\usepackage{xifthen}

\usepackage{comment}

\definecolor{caribbeangreen}{rgb}{0.0, 0.8, 0.6}

\usepackage{ulem}

\def\eps{\varepsilon}
\def\bR{{\mathbb R}}
\def\bN{{\mathbb N}}
\def\bZ{{\mathbb Z}}
\def\bX{{\mathbb X}}

\def\cM{\mathcal{M}}
\def\cL{\mathscr{L}}
\def\textgirth{bundledness}

\def\df{\mathrm{d}_\mathcal{F}}
\def\ddf{\mathrm{d}_{\mathrm{d}\mathcal{F}}}

\def\ball{\mathrm{B}}
\def\dd{\mathrm{d}}

\newcommand{\xx}[2]{\bX^{#1,#2}}
\newcommand{\xl}[3]{\bX^{#1,#2}_{#3}}

\title{$(1+\eps)$-ANN Data Structure for Curves via Subspaces of Bounded Doubling Dimension}
\author{Jacobus Conradi}{Department of Computer Science, University of Bonn, Germany}{conradi@cs.uni-bonn.de}{https://orcid.org/0000-0002-8259-1187}{Partially funded by the Deutsche Forschungsgemeinschaft (DFG, German Research Foundation) - 313421352 (FOR 2535 Anticipating Human Behavior). Affiliated with Lamarr Institute for Machine Learning and Artificial Intelligence.}
\author{Anne Driemel}{Hausdorff Center for Mathematics, University of Bonn, Germany}{driemel@cs.uni-bonn.de}{https://orcid.org/0000-0002-1943-2589}{Affiliated with Lamarr Institute for Machine Learning and Artificial Intelligence.}
\author{Benedikt Kolbe}{Hausdorff Center for Mathematics, University of Bonn, Germany}{bkolbe@uni-bonn.de}{}{}

\authorrunning{Conradi, Driemel, and Kolbe}

\Copyright{Jacobus Conradi, Anne Driemel, Benedikt Kolbe} 


\nolinenumbers

\begin{document}
\hideLIPIcs

\maketitle

\begin{abstract}
We consider the $(1+\eps)$-Approximate Nearest Neighbour (ANN) Problem for polygonal curves in $d$-dimensional space under the Fréchet distance and ask to what extent known data structures for doubling spaces can be applied to this problem. Initially, this approach does not seem viable, since the doubling dimension of the target space is known to be unbounded --- even for well-behaved polygonal curves of constant complexity in one dimension. In order to overcome this, we identify a subspace of curves which has bounded doubling dimension and small Gromov-Hausdorff distance to the target space. 
We then apply state-of-the-art techniques for doubling spaces and show how to obtain a data structure for the $(1+\eps)$-ANN problem for any set of parametrized polygonal curves. The expected preprocessing time needed to construct the data-structure is $F(d,k,S,\eps)n\log n$ and the space used is $F(d,k,S,\eps)n$, with a query time of $F(d,k,S,\eps)\log n + F(d,k,S,\eps)^{-\log(\eps)}$, where $F(d,k,S,\eps)=O\left(2^{O(d)}k\Phi(S)\eps^{-1}\right)^k$ and $\Phi(S)$ denotes the spread of the set of vertices and edges of the curves in $S$. We extend these results to the realistic class of $c$-packed curves and show improved bounds for small values of $c$. 


\end{abstract}

\newpage
\setcounter{page}{1}


\section{Introduction}
Given a set $S$ of $n$ points, the Nearest Neighbour Problem is the problem of finding the point in $S$ that minimizes the distance to a given query point $q$.
The Nearest Neighbour Problem is a fundamental problem whose variants have long been studied and applied in different areas, such as RNA sequencing \cite{bindewaldRNA}, disease diagnosing \cite{shoumanApplyingkNN}, motion pattern detection \cite{gudmundssonEfficient}, shape indexing \cite{beisShape} or handwritten digit recognition \cite{leeHandwritten}. The problem has been studied as early as the 1960s \cite{minskyPerceptrons}, and classical results such as the one by Shamos via point location in a Voronoi diagram achieve a query time of $O(\log n)$ while using $O(n\log n)$ space in $\bR^2$ \cite{shamosClosest}, which was later improved upon by Kirkpatrick to only require linear space and preprocessing time with the logarithmic query time~\cite{Kirkpatrick1983Optimal}. 

Various methods have been employed to design data structures for approximate solutions to the Nearest Neighbour Problem \cite{arya1998optimal,harpeledApproximateNN,kushilevitzEfficient}. However, these approaches often require some sense of `Euclidean' dimension, while applications often require to work in spaces with far more complicated distance measures. One complexity measure often used to generalize many different results to complicated metric spaces is the notion of doubling dimension~\cite{guptaBounded,kargerFinding,talwarBypassing}. The doubling dimension is the smallest number $d$ such that any metric ball inside the metric space can be covered by $2^d$ many balls of half the radius. It is a well known fact that in Euclidean space the doubling dimension roughly corresponds to the dimension, that is, the doubling dimension of $\bR^d$ is $\Theta(d)$. The Approximate Nearest Neighbour (ANN) Problem in spaces with low doubling dimension has been studied extensively~\cite{kargerFinding} and results are known that roughly match the bounds known for $\bR^d$ \cite{harpeledFast}.

The metric space we are interested in is the space of polygonal curves in $\bR^d$ under the Fréchet distance. Polygonal curves naturally arise from any sort of motion tracking, such as GPS data or motion capture data, and are therefore of much interest. The Fréchet distance is a natural distance measure on parametrized curves~\cite{Alt1995ComputingtFdbTPC} that --- unlike the Hausdorff distance --- takes into account the parametrization of the curves and has received considerable attention~\cite{ agarwal2014computing,aronov2006frechet, buchin2017four, chan2018improved, Colombe2020ApproximatingT, Driemel2010ApproximatingTF}. The metric space of curves under the Fréchet distance has been shown to have unbounded doubling dimension~\cite{driemelClustering}, which suggests that data structures designed for doubling spaces would perform poorly. Our primary motivation in this paper is to find a workaround to this problem, and to leverage the rich background of ANN results for doubling spaces to the Fréchet distance after all.



\subsection{Basic definitions}
The metric space we study in this paper is the set of all polygonal curves with the Fréchet distance as its metric.
\begin{definition}[polygonal curve]
    An \textbf{edge} in $\bR^d$ is the continuous map obtained from the linear interpolation of two points $a$ and $b$ in $\bR^d$. We may write $\overline{a\,b}$ to denote this unique edge.
	A \textbf{polygonal curve} $T:[0,1]\to \bR^d$ of complexity $n$ is defined by an ordered set of $n$ points in $\bR^d$ and is the result of $(n-2)$ concatenations of the $(n-1)$ edges in $\bR^d$ defined by any two consecutive points.
	We call the underlying points of a polygonal curve of complexity $n$ its \textbf{vertices}.
	For $0\leq s \leq t\leq 1$ we denote the subcurve of $T$ from $T(s)$ to $T(t)$ by $T[s,t]$.
\end{definition}

\begin{definition}[Fréchet distance]\label{def:frechet} 
	Given two curves $X$ and $Y$ in $\bR^d$, their Fréchet distance  is defined as
	\[ \df(X,Y) = \inf_{f,g:[0,1]\rightarrow[0,1]} ~ \max_{t\in[0,1]}\|X(f(t)) - Y(g(t))\|\]
	where $f$ and $g$ are continuous, non-decreasing and surjective.
\end{definition}

In our definition of the Fr\'echet distance given above, we follow Alt and Godau~\cite{alt1995approximate}.
The Fréchet distance is generally not a metric, but rather a pseudo-metric, as there are curves $X\neq Y$, such that $\df(X,Y)=0$. However this is easily remedied by considering the quotient space induced by the equivalence relation $X\sim Y \iff \df(X,Y)=0$. If we assume that no vertex of an input curve lies inside the convex hull of its two neighbors along the curve, then the ordered sequence of vertices of two equivalent curves is the same. Therefore, we can eliminate all duplicates in near-linear time using lexicographical sorting. This results in a set of curves with pairwise non-zero distance. In the following, we simply assume that the Fr\'echet distance between any two curves in the input is non-zero.

\begin{definition}
    Denote by $(\xx{d}{k},\df)$ the metric space of polygonal curves in $\bR^d$ with complexity $k$ under the continuous Fréchet distance. We further write  $(\xl{d}{k}{\Lambda},\df)$ for the subspace of polygonal curves in $(\xx{d}{k},\df)$ where the length of each edge is bounded by $\Lambda$.
\end{definition}

A well-studied variant of the continuous Fréchet distance is the discrete Fréchet distance where only the vertices---and not their connecting edges---are considered in the computation. Similar to the continuous Fréchet distance, the discrete Fréchet distance is a pseudo-metric and can be regarded as a metric in a corresponding quotient space.
\begin{definition}[discrete Fréchet distance]
	Given two polygonal curves $X$ and $Y$ in $\bR^d$ defined by vertices $x_1,\ldots,x_n$ and $y_1,\ldots,y_m$, their discrete Fréchet distance is defined as
	\[ d_{d\mathcal{F}}(X,Y) = \inf_{f,g} ~ \max_{t\in[0,1]}\|x_{f(t)} - y_{g(t)}\|,\]
	where $f:[0,1]\rightarrow\{1,\ldots,n\}$ and $g:[0,1]\rightarrow\{1,\ldots,m\}$ are non-decreasing and surjective functions.
\end{definition}

\begin{observation}\label{obs:dfddf}
Let $P,Q$ be polygonal curves in $\bR^d$. Then $\df(P,Q)\leq\ddf(P,Q)$.
\end{observation}


\subparagraph{Problem definition ($(1+\eps)$-ANN)}
Let $(\cM,\dd_\cM)$ be a metric space. Let $P\subset \cM$ be a set of points in $\cM$ and a parameter $\eps>0$ be given. For a given point $q\in\cM$, the $(1+\eps)$-Approximate Nearest Neighbour Problem~($(1+\eps)$-ANN) is to find a point $\hat{x}\in P$ whose distance to $q$ approximates the distance to the nearest neighbour in $P$. Specifically, $\hat{x}\in P$ is a valid solution iff for all $x\in P$ it holds that
\[\dd_\cM(q,\hat{x})\leq(1+\eps)\,\dd_\cM(q,x)\]

\subsection{State of the art}\label{sec:previouswork}

Most of the work on nearest-neighbour data structures for curves has focused on the discrete Fréchet distance. 
One such recent result is due to Filtser, Filtser and Katz~\cite{filtser2019approximate} who presented a data structure for the $(1+\eps,r)$-Approximate Near Neighbour Problem of size $n\cdot O(1/\eps)^{md}$, with a query time of $O(md)$, where $n$ is the number of input curves, and $m$ is the complexity of the input-curves. When the complexity $k$ of the query curve is small compared to the complexity $m$ of the input curves, the space can be improved to $n\cdot O(1/\eps)^{kd}$ with query time $O(kd\log(nkd/\eps))$. In the $(1+\eps,r)$-Approximate Near Neighbour Problem the goal is to construct a data structure on a set of input-curves which for a given query-curve outputs any of the input-curves that is at distance at most $(1+\eps)r$ to the query curve, if there is an input-curve with distance at most $r$ to the query-curve. 

Results for the $(1+\eps,r)$-Approximate Near Neighbour problem readily extend to the $(1+\eps)$-Approximate \textit{Nearest} Neighbour Problem by a result of Har-Peled, Indyk and Motwani~\cite{harpeledApproximateNN}. This reduction incurs merely an additional logarithmic factor $O(\log^2(n))$ in the size, and $O(\log n)$ in the query time.

In contrast to the multitude of approaches to the discrete Fréchet distance for arbitrary dimension, results concerning the continuous Fréchet distance appear harder to come by.
Consider the naive approach of approximating the continuous Fréchet distance via the discrete Fréchet distance. For this, let a set of $n$ curves in $\bX^{d,k}_\Lambda$ as well as the approximation parameter $\eps$ be given. For the discrete Fréchet distance to approximate the continuous Fréchet distance up to an additive term of $r\eps$ for some $r>0$, we require successive vertices to lie at most $\Theta(r\varepsilon)$ far apart. Thus we subdivide every edge into edges of length at most $r\eps$, resulting in a set of $n$ curves each of complexity $O(k\Lambda/(r\eps))$. Building the $(1+\eps,r)$-ANN data structure from~\cite{filtser2019approximate} results in a space requirement of $n\cdot O(1/\eps)^{O(dk\Lambda/(r\eps))}$. As the radii used in the reduction can be as small as $r^*/n$, where $r^*$ is roughly the distance of some two input curves, this extends to a data structure of size $O(n\log^2(n))\cdot O(1/\eps)^{O(ndk\Lambda/(r^*\eps))}$ for the $(1+\eps)$-ANN problem.

Bringmann et al.~\cite{bringmann2022tight}  showed that there is a $(1+\eps,r)$-Approximate Near Neighbour data structure for the continuous Fréchet distance in one dimension, which uses $n\cdot O(\frac{m}{k\eps})^k$ space, needs $O(nm)\cdot O(\frac{m}{k\eps})^k$ expected preprocessing time and achieves a query time of $O(k\cdot2^k)$. 
They also show tightness of their data structure bounds in several scenarios. More precisely, they show conditional lower bounds  based on the Orthogonal Vectors Hypothesis that give reason to believe that one cannot achieve preprocessing time $\mathrm{poly}(n)$ and query time $O(n^{1-\eps'})$ at the same time, when $k$ is $1\ll k \ll \log n$ and $m > k\cdot  n^{c/k}$, even if $d=1$. Their arguments also apply to the $c$-ANN problem under the continuous Fr\'echet distance for any $c \in [1,2)$.  

In two dimensions, Afshani and Driemel presented a data structure based on semi-algebraic range searching that solves the (exact) Near Neighbour Problem under the Fr\'echet distance~\cite{afshani2018complexity}. The data structure needs $O(n(\log\log n)^{O(m^2)})$ space and has query time $O(\sqrt{n}\log^{O(m^2)}n)$. 

In higher dimensions, Mirzanezhad recently presented a data structure result for the $(1+\eps,r)$-Approximate Near Neighbour Problem under the continuous Fréchet distance, using space in $n\cdot O((\max(1,D)\sqrt{d}/\eps^2)^{kd})$  and query time in $O(kd)$, where $D$ denotes the diameter of the underlying vertex set of the input curves~\cite{MIRZANEZHAD2023journal}. 
However, as presented this data structure works only if $\eps<\delta$ (refer to Theorem 8 in the arXiv version~\cite{MirzanezhadANN}). 
The data structure covers the entirety of the input-curves with a grid of edge-length roughly $\eps \delta$, to precompute an answer for every sequence of $k$ gridpoints. A query-curve is then snapped to the closest grid points and the precomputed answer is given as the output. 
For smaller values of $\delta$ one would have to scale the input, increasing $D$ accordingly. As a result, combining this data structure with the standard reduction by Har-Peled, Indyk and Motwani~\cite{harpeledApproximateNN} does not lead to an efficient data structure for the ANN-problem. 

Independent to our work, Cheng and Huang very recently presented a $(1+\eps)$-ANN data structure for polygonal curves in arbitrary dimension under the continuous Fréchet distance. The data structure uses space in  $\tilde{O}\left(k(mnd^d/\eps^d)^{O(k+1/\eps^2)}\right)$ and achieves query time in $\tilde{O}\left(k(mn)^{0.5+\eps}/\eps^d+k(d/\eps)^{O(dk)}\right)$. Here, $\tilde{O}(\cdot)$ hides polylogarithmic factors with exponents in $O(d)$.


\subsection{Our contribution}

In this paper, we provide a $(1+\eps)$-ANN data structure for a set of curves in arbitrary dimensions under the continuous Fréchet distance. One of the main ingredients is the construction of a suitable subspace of the space of all polygonal curves. This subspace has small distance to the original space as metric spaces and additionally, unlike the space of all polygonal curves, bounded doubling dimension. Our approach can thus be seen as a particular instance of the more general idea of using results on the ANN problem for spaces of bounded doubling dimension in the context of spaces that are close to a space of bounded doubling dimension, even if they do not themselves have this property.
This approach was inspired by the work of Sheehy and Sheth on the bottleneck distance for persistence diagrams~\cite{sheehyNearlydsopd}.

Throughout this paper, for simplicity of the exposition, we will not distinguish between small or large $k$, that is, we assume $k=m$. A key result is the following (Section~\ref{sec:ANN}).

\begin{restatable}{theorem}{lemapproxresult}\label{lem:approxresult}
Given a set $S$ of $n$ polygonal curves in $\xl{d}{k}{\Lambda}$ and parameters $0<\eps< 1$ and $\eps'>0$, one can construct a data structure that for given $q\in\xx{d}{k}$ outputs an element $s^*\in S$ such that for all $s\in S$ it holds that $\df(s^*,q)\leq(1+\eps)\df(s,q) + \eps'$. The query time is $F(d,k,\Lambda,\eps')\log n + F(d,k,\Lambda,\eps')^{-\log(\eps)}$, the expected preprocessing time is $F(d,k,\Lambda,\eps')n\log n$ and the space used is $F(d,k,\Lambda,\eps)n$, where $F(d,k,\Lambda,\eps') = O\left(2^{O(d)}k(1+\Lambda/\eps')\right)^k$.
\end{restatable}

To turn this into a data structure with a purely multiplicative error we choose $\eps'$ as the smallest Fréchet distance of any two distinct input curves over $\eps$. Thus, the running time and space complexity of our data structure polynomially depends on a numerical value, which we call the \textgirth{} of the set of input curves.

\begin{definition}[\textgirth{}]
Given a set of curves $S\in\xx{d}{k}$, the \textgirth{} $\mathcal{G}(S)$ of $S$ is defined as 
\[\mathcal{G}(S) = \frac{\min_{s\neq s' \in S}\df(s,s')}{\max_{e\in E(S)}\|e\|},\]
where $E(S)$ denotes the set of edges of curves in $S$ and $\|e\|$ denotes the length of the edge $e$.
\end{definition}

The \textgirth{} is reminiscent of the global stretch --- a measure of complexity on geometric graphs --- which was introduced by Erickson~\cite{erickson2005localpolyhedra} and further analyzed by Bose et al.~\cite{bose2010sigmalocal}. The \textgirth{} is closely related to the spread of the set of vertices and edges of the input curves, in that it is lower bounded by the reciprocal of the spread (see Lemma~\ref{lem:girthspread}).

\begin{definition}[spread]
For a point set $P$ in some metric space $(\cM,d_\cM)$, we define the spread $\Phi(P)$ as the ratio between the maximal and minimal pairwise distance of points in $P$. Similarly, define the spread $\Phi(S)$ of a collection of sets as the ratio between the maximal and minimal non-zero pairwise distances of sets in $S$, where the distance between two sets $A,B\subset\cM$ is defined as $\dd_\cM(A,B)=\min_{a\in A}\min_{b\in B}\dd_\cM(a,b)$.
\end{definition}

With these definitions, our principal results can be summarized as follows (Section~\ref{sec:ANN}).


\begin{restatable}{theorem}{thmgirthmain}\label{thm:girth_main}
Given a set $S$ of $n$ polygonal curves in $\xx{d}{k}$ and $0<\eps\leq 1$ one can construct a data structure answering $(1+\eps)$-approximate nearest neighbour queries. The query time is $F(d,k,S,\eps)\log n + F(d,k,S,\eps)^{-\log(\eps)}$, the expected preprocessing time is $F(d,k,S,\eps)n\log n$ and the space used is $F(d,k,S,\eps)n$, where $F(d,k,S,\eps)=O\left(2^{O(d)}k(1+\mathcal{G}(S)^{-1}\eps^{-1})\right)^k$.
\end{restatable}


Replacing the \textgirth{} with the more pessimistic spread of the vertices and edges of the input curves we get the following result.

\begin{restatable}{corollary}{corgirthmain}
Given a set $S$ of $n$ polygonal curves in $\xx{d}{k}$ and $0<\eps\leq 1$ one can construct a data structure answering $(1+\eps)$-approximate nearest neighbour queries. The query time is $F(d,k,S,\eps)\log n + F(d,k,S,\eps)^{-\log(\eps)}$, the expected preprocessing time is $F(d,k,S,\eps)n\log n$ and the space used is $F(d,k,S,\eps)n$, where $F(d,k,S,\eps)=O\left(2^{O(d)}k\Phi(S)\eps^{-1}\right)^k$, where $\Phi(S)$ denotes the spread of the set of vertices and edges of the curves in $S$.
\end{restatable}


In the special case that all curves in $S$ are $c$-packed for some constant $c>0$, that is the length of the intersection of the curve with any ball is at most $c$ times the balls radius, the parameter $F(d,k,S,\eps)$ is instead in $O\left(2^{O(d)}(1+\mathcal{G}(S)^{-1}\eps^{-1})\right)^k$, or, more pessimistically, in $O\left(2^{O(d)}\Phi(S)\eps^{-1}\right)^k$.

In Section~\ref{sec:generalization} we also discuss sufficient conditions for our approach to generalize to other settings: when a space of unbounded doubling dimension is approximated by a well-suited family of subspaces of bounded doubling dimension.

\subsection{Technical overview}

In Section~\ref{sec:spaces}, we introduce a subspace $\bX^*$ of $(\bX^{d,k},\df)$  and show how the subspace relates to $(\bX^{d,k},\df)$. A significant part of our work  is concerned with analyzing properties of this subspace $\bX^*$.
In Section~\ref{sec:bounded}, we present the analysis of the upper bound on the doubling constant and dimension of $\bX^*$. We furthermore extend this analysis to the special case that the curves in the ambient space are $c$-packed, resulting in an improvement in the upper bound. 
In Section~\ref{sec:lowerbound}, we extend the lower bound construction from~\cite{driemelClustering} to argue that our bounds on the doubling dimension of $\bX^*$ are almost tight.
In Section~\ref{sec:ANN}, we then use the bound on the doubling dimension of $\bX^*$ to construct an $(1+\eps)$-ANN data structure based on the work of Har-Peled and Mendel~\cite{harpeledFast}. In Section~\ref{sec:generalization}, we then generalize these results to present a general framework, which specifies how one can extend $(1+\eps)$-ANN data structures for spaces of bounded doubling dimension to spaces of unbounded doubling dimension.

\section{Curve spaces}\label{sec:spaces}

\begin{definition}[doubling constant and dimension]
Let $(\cM,\dd_\cM)$ be a metric space. Define the \textbf{$r$-ball} in $\cM$ by $\ball^{\cM}_r(x)=\{y\in \cM\mid \dd_\cM(x,y)\leq r\}$, as the closed ball of radius $r>0$ centered at $x$. The \textbf{doubling constant} of $(\cM,\dd_\cM)$ is defined as the minimal number $\nu$, such that for any $x\in \cM$ and any $r>0$ the ball $\ball_r^\cM(x)$ of radius $r$ is contained in the union of at most $\nu$ balls with radius $r/2$. The \textbf{doubling dimension} of $(\cM,\dd_\cM)$ is defined as $\log_2(\nu)$.
\end{definition}
We may omit the metric space in the notation of a ball whenever the metric space is clear, writing $\ball_r(x)$ instead of $\ball^\cM_r(x)$ for $x\in\cM$.

It turns out that for $k\geq 3$ the doubling dimension of $(\xx{d}{k},\df)$ is unbounded~\cite{driemelClustering}. A straightforward modification to their construction yields the following theorem.
\begin{theorem}\label{thm:doublingunbounded}
The doubling constant of $(\xl{d}{k}{\Lambda},\df)$ is unbounded for any $k\geq 3$ and $\Lambda>0$.
\end{theorem}
Theorem~\ref{thm:doublingunbounded} motivates the search for a subspace $(\bX^*,\df)$ of $(\xl{d}{k}{\Lambda},\df)$ of bounded doubling dimension. To answer $(1+\eps)$-ANN queries in $\xl{d}{k}{\Lambda}$, the ansatz is to map the input curves to $(\bX^*,\df)$ and answer queries in this subspace. With this approach, the quality guarantee of the $(1+\eps)$-ANN queries decreases by an additional additive factor, roughly depending on the distortion of the map between $(\xl{d}{k}{\Lambda},\df)$ and $(\bX^*,\df)$. 

\begin{definition}[$(\mu,\eps)$-curves]
For any $\eps>0$ and $\mu\in\bN$ define the space of $(\mu,\eps)$-curves in $\bX^{d,k}$ as the subspace of $(\xx{d}{k},\df)$ induced by the set of polygonal curves in $\bX^{d,k}$ whose edge lengths are all exact multiples of $\eps$. We further require the edge lengths to be bounded by $\mu\eps$. The space of $(\mu,\eps)$-curves in $\bX^{d,k}$ is a natural subspace of $(\xl{d}{k}{\mu\eps},\df)$.
\end{definition}

We may abuse notation slightly, and not specify the ambient space $\xx{d}{k}$ of the space of $(\mu,\eps)$-curves, if the ambient space is clear. We may likewise write $\cM$ when talking about the metric space $(\cM,\dd_\cM)$, if the metric is clear.

\begin{figure}
    \centering
    \includegraphics[width=\textwidth]{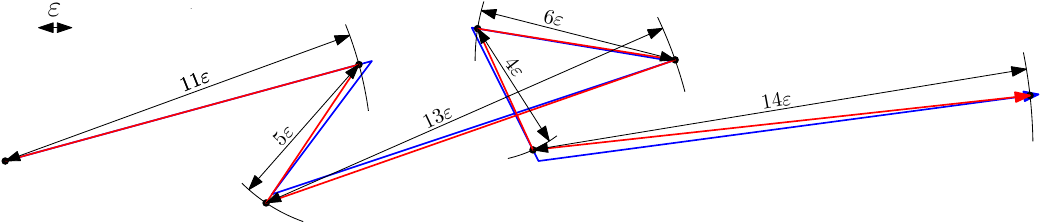}
    \caption{Example of a curve $P\in\xl{d}{k}{\Lambda}$ in blue, and an $\eps$-curve close to $P$ resulting from Lemma~\ref{lem:simplification} in red.}
    \label{fig:simplification}
\end{figure}

\begin{lemma}\label{lem:simplification}
Let $P\in\xl{d}{k}{\Lambda}$ be a polygonal curve and  $\eps>0$. We can construct a $({\lceil\Lambda/\eps\rceil+1},{\eps})$-curve $P'$ in $ \xx{d}{k}$ such that $\df(P,P')\leq\eps/2$ in $O(k\log(\Lambda/\eps))$ time.
\end{lemma}
\begin{proof}
Let $p_1,\ldots,p_k$ be the vertices of $P$ and $P'=\emptyset$. We begin by adding $p'_1=p_1$ to $P'$. Now assume $p'_{i-1}$ is the last vertex of $P'$. Compute the value $\mu_i$, such that the magnitude $|\mu_i\eps-\|p_i-p'_{i-1}\||$ is minimal. Then add $p'_i = p'_{i-1}+(\mu_i\eps)\frac{p_i-p'_{i-1}}{\|p_i-p'_{i-1}\|}$ to $P'$. Note that by construction $\|p_i-p'_{i}\|\leq\eps/2$. Hence $d_F(P,P')\leq\eps/2$. The length of the edges of $P'$ are bounded by $(\lceil\Lambda/\eps\rceil+1)\eps$. Indeed, $\|p_i-p'_{i}\|\leq\eps/2$ and $\|p_i-p_{i-1}\|\leq\Lambda$ imply that  $\|p'_i-p'_{i-1}\|\leq\Lambda+\eps$. For the running time, for every $i\leq n$ the value $\mu_i$ such that the magnitude $|\mu_i\eps-\|p_i-p'_{i-1}\||$ is minimal can be identified in $O(\log(\Lambda/\eps))$ time. This is done by first identifying the smallest power of $2$ larger than $\|p_i-p'_{i-1}\|/\eps$ and then binary searching over the integer multiples of $\eps$ up to this power of $2$. As we do this once for every edge, the claimed running time follows.
\end{proof}

The additive error incurred by mapping $\xl{d}{k}{\Lambda}$ to the space of $(M,\eps)$-curves in $\xx{d}{k}$ for $M\geq \lceil\Lambda/\eps\rceil+1$ depends only on the distortion of the map. The Gromov-Hausdorff distance is a related measure of the distance between metric spaces. For $\mu\rightarrow\infty$ the space of $(\mu,\eps)$-curves in $\xx{d}{k}$ will have (depending on $\eps$) small Gromov-Hausdorff distance to the ambient space of curves. We map any element from $\xx{d}{k}$ into this subspace by rounding the length of every edge to a multiple of $\eps$, thereby moving the vertices by a small amount each time (refer to Figure \ref{fig:simplification}).

\begin{definition}[Gromov-Hausdorff distance]
The Gromov-Hausdorff distance is a distance measure on metric spaces. Let $\cM$ and $\mathcal{N}$ be two metric spaces. Then the Gromov-Hausdorff distance is defined as
\[\mathrm{d}_{GH}(\cM,\mathcal{N}) = \inf_\mathcal{Z}\{\mathrm{d}^{\mathcal{Z}}_H(f(\cM),g(\mathcal{N}))\mid f:\cM\rightarrow\mathcal{Z},g:\mathcal{N}\rightarrow\mathcal{Z}\text{ isometric embeddings}\},\]
where $\mathcal{Z}$ ranges over metric spaces and $\mathrm{d}^{\mathcal{Z}}_H(X,Y)$ denotes the Haus\-dorff distance of two sets in the metric space $\mathcal{Z}$ which is defined as $\max\{\sup_{x\in X}\inf_{y\in Y}d_{\mathcal{Z}}(x,y),\sup_{y\in Y}\inf_{x\in X}d_{\mathcal{Z}}(x,y)\}$. 
\end{definition}

\begin{corollary}
Let $d,k$ and $\eps>0$ be given, $\mu\rightarrow\infty$, and denote by $C_\mu$ the space of $(\mu,\eps)$-curves in $\xx{d}{k}$. Then $\mathrm{d}_{GH}(\xx{d}{k},C_\mu)\rightarrow\eps/2$.
\end{corollary}
\begin{proof}
For every $P\in\xx{d}{k}$ and $\eps>0$ there is a finite integer $M<\infty$, such that $P\in\xl{d}{k}{M\eps}$. Then by Lemma \ref{lem:simplification} there is a $P'\in C_{M+1}$ with $\df(P,P')\leq\eps/2$. As $C_\mu$ is a natural subspace of $\xx{d}{k}$ for all $\mu\in\bN$, the claim holds.
\end{proof}




\section{Bounding the doubling dimension of the space of  $(\mu,\eps)$-curves}\label{sec:bounded}

In this section, we study the doubling dimension of the space of $(\mu,\eps)$-curves in $\xx{d}{k}$. 
Unfortunately, our bound is non-constructive. As such, it does not provide a doubling oracle that, for a given ball of radius $r$ in the metric space, outputs a set of balls of radius $r/2$ which cover the ball of radius $r$.

\subsection{Properties of the Euclidean space}

Before diving into the analysis of the doubling dimension of the space of $(\mu,\eps)$-curves, we begin by analyzing properties of the ambient space $\bR^d$. For this we often inspect so called $\Delta$-neighbourhoods of subsets of $\bR^d$. For any subset $A$ the \textbf{$\Delta$-neighbourhood} of $A$ is defined by $N_\Delta(A)=\{x\in \bR^d \mid \exists a\in A : d(x,a)\leq\Delta\}$. Note that the $\Delta$-neighbourhood of a single point $x$ coincides with a ball of radius $\Delta$ centered at $x$.

It is well-known that the doubling dimension of $\bR^d$ under the Euclidean norm is in $\Theta(d)$.
We prove a more general statement which relates the volume of some arbitrary set to the amount of balls needed to cover it.
For this, we denote by $\lambda^d$ the Lebesgue measure in $\bR^d$ and by $\lambda^d(A)$ the \textbf{volume} of a (measurable) set $A\subset\bR^d$. 

\begin{lemma}\label{lem:rdcover}
Let $A\subset\bR^d$ be a bounded set, and let $r>0$ be fixed. Then there is a set of points $C\subset A$ of cardinality $\lceil\lambda^d(N_{r/2}(A))/V^d_{r/2}\rceil$ such that $A\subset\bigcup_{c\in C}\ball_{r}(c)$, where $V^d_{r/2}$ denotes the volume of a $d$-dimensional ball of radius $r/2$.
\end{lemma}
\begin{proof}
We construct the set $C$ greedily. For this we start with $C=\emptyset$, and then iteratively add any point from $A\setminus\bigcup_{c\in C}\ball_{r}(c)$, until $A\subset\bigcup_{c\in C}\ball_{r}(c)$. As any two points in $C$ have a distance of at least $r$, balls centered at points of $C$ with radius $r/2$ are disjoint and clearly contained in $N_{r/2}(A)$. Thus the number of points in $C$ is bound by $\lceil\lambda^d(N_{r/2}(A))/V^d_{r/2}\rceil$.
\end{proof}
\begin{lemma}\label{lem:circlecover}
For any $r>0$ and $c>1$, any ball $\ball_r(p)\subset\bR^d$ can be covered by $O((2c+1)^d)$ balls of radius $r/c$. 
\end{lemma}
\begin{proof}
This follows from the classical result that 
$V^d_r=\frac{\pi^{d/2}}{\Gamma(d/2+1)}r^d$, where $\Gamma$ denotes the Gamma-function \cite{parksTheVotUnB}. As $N_{r/2c}(\ball_r(p)) = \ball_{r(2c+1)/2c}(p)$, Lemma~\ref{lem:rdcover} implies that any $\ball_r(p)$ can be covered with $\lceil(r(2c+1)/2c)^d/(r/2c)^d\rceil = O((2c+1)^d)$ balls of radius $r/c$.
\end{proof}

Later on, we will use this lemma to retrieve a set of points, namely the centers of the balls used in such a covering, as candidates for vertices of curves that will be centers of balls in $\bX^{d,k}$.

When analyzing the doubling dimension, we often consider what an edge might look like that has Fréchet distance at most $\Delta>0$ to a subcurve of an input curve. The basic tool we use for this analysis is the observation that the edge under consideration then is a $\Delta$-stabber of the vertices (and indeed any ordered set of points along the subcurve) of the subcurve.

\begin{definition}[$\Delta$-stabber]
Let an ordered set of points $(p_1,\ldots,p_n)$ in $\bR^d$ be given. A polygonal curve $l:[0,1]\rightarrow\bR^d$ between two points $a,b\in\bR^d$ is called a $\Delta$-stabber of $(p_1,\ldots,p_n)$ if there are values $0\leq t_1\leq \ldots \leq t_n\leq 1$ such that $\|l(t_i)-p_i\|\leq\Delta$ for all $1\leq i\leq n$.
\end{definition}

The notion of $\Delta$-stabbers has been introduced by Guibas et al.~\cite{Guibas1994ApproximatingPSM}, and is closely related to the Fréchet distance. Any edge that has Fréchet distance at most $\Delta$ to a polygonal curve defined by vertices $(p_1,\ldots)$ is a $\Delta$-stabber of the ordered point set $(p_1,\ldots)$. Similarly, a $\Delta$-stabber of the ordered point set $(p_1,\ldots)$ contains an edge that has Fréchet distance at most $\Delta$ to the polygonal curve defined by the vertices $(p_1,\ldots)$.


\begin{observation}\label{obs:stabber}
Let a polygonal curve $P$ and $\Delta>0$ be given. Let $e=\overline{p\,q}$ be an edge such that for given $0\leq s\leq t\leq1$ the Fréchet distance $\df(P[s,t],e)$ is at most $\Delta$. Then for any $s\leq m\leq t$ the edge $e$ is a $\Delta$-stabber of $(P(s),P(m),P(t))$.
\end{observation}

\subsection{Packing the metric ball}

The result we now want to prove is the following theorem. Proving it constitutes the main bulk of work in this section.

\begin{theorem}\label{thm:main}
Let $k,\mu,d\in\bN$ and $\eps>0$. 
The doubling constant of the space of $(\mu,\eps)$-curves in $\xx{d}{k}$ is bounded by $O(43^dk\mu)^k$ and thus the doubling dimension of the space of $(\mu,\eps)$-curves is bounded by $O(k(d+\log(k\mu)))$.
\end{theorem}

Let $P$ be a $(\mu,\eps)$-curve in $\xx{d}{k}$. Our objective is to cover the $\Delta$-neighbourhood of $P$ with respect to $\df$ with balls of radius $\Delta/2$. We encounter the question of where, given $p\in\bR^d$, we may place a second point $q$ such that there is a subcurve of $P$ that is close to $\overline{p\,q}$. Indeed, for any curve $Q$ with $\df(P,Q)\leq\Delta$, the endpoint of any edge of $Q$ is potentially a candidate for such a $q$, depending only on the starting point of the edge in question.

\begin{definition}
Let $P$ be a polygonal curve and $\lambda\geq 0$ and $\Delta\geq 0$. For $s\in[0,1]$ define the locus of edge endpoints of edges close to subcurves of $P$ starting at the parameter $s$ as the set
\[L_{\lambda,\Delta}(P,s)=\left\{q\in\bR^d\middle|\exists p\in\bR^d \text{ and } \exists t\in[s,1] \text{ with } \|p-q\|=\lambda \text{, }\df(P[s,t],\overline{p\,q})\leq\Delta\right\}.\]
\end{definition}

The following lemma motivates discretising the lengths of edges, as for a fixed length $\lambda$ the set $L_{\lambda,\Delta}(P,s)$ is contained in a single ball of constant size. 

\begin{figure}
    \centering
    \includegraphics[width=\textwidth]{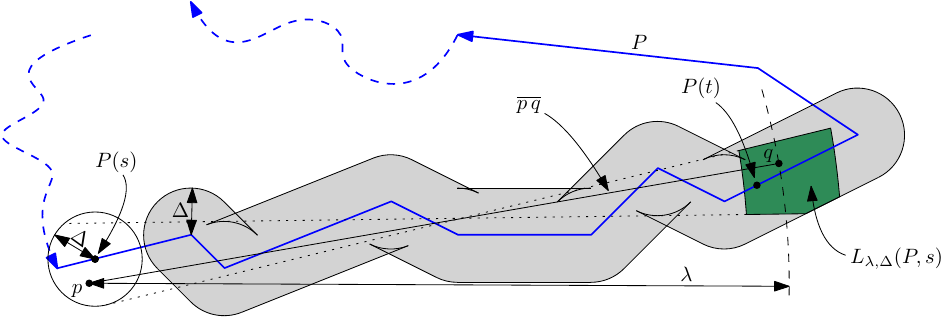}
    \caption{Illustration of the set $L_{\lambda,\Delta}(P,s)$ in dark green, together with the points $p$ and $P(t)$ realizing a point $q$ in $L_{\lambda,\Delta}(P,s)$, that is $\|p-q\|=\lambda$ and $\df(P[s,t],\overline{p\,q})\leq\Delta$.}
    \label{fig:LDendpoints}
\end{figure}

\begin{lemma}\label{lem:psld}
Let $P$ be a polygonal curve. Let $\lambda\geq 0$ and $\Delta\geq0$ be given. Then for every $s\in[0,1]$ there is a point $p^*\in\bR^d$, such that
\[L_{\lambda,\Delta}(P,s)\subset\ball_{5\Delta}(p^*).\]
\end{lemma}
\begin{proof}
Assume $L_{\lambda,\Delta}(P,s)$ is nonempty, as otherwise we are done. Similarly assume $\lambda \geq 4\Delta$ as otherwise $L_{\lambda,\Delta}(P,s)$ is trivially contained in $\ball_{\lambda+\Delta}(P(s))\subset\ball_{5\Delta}(P(s))$. Now let $t^*\geq s$ be the smallest value such that for the corresponding points along $P$ it holds that $\|P(s)-P(t^*)\|\geq \lambda - 2\Delta$. Then we claim that the sought-after point is $p^*=P(t^*)$.\\
For the remainder of this proof refer to Figure~\ref{fig:psld}. Let $q\in L_{\lambda,\Delta}(P,s)$ be given. Then by definition there is a point $p$ and a value $t$, such that $\overline{p\,q}$ has length $\lambda$ and $\df(P[s,t],\overline{p\,q})\leq\Delta$. Thus $\|P(s)-p\|\leq\Delta$ and similarly $\|P(t)-q\|\leq\Delta$. But then by the triangle inequality $\|P(t)-P(s)\|\geq\lambda - 2\Delta$. Thus $t\geq t^*$. Hence, $\overline{p\,q}$ is an $\Delta$-stabber of $(P(s),P(t^*),P(t))$. This implies that there is a point $m$ along $\overline{p\,q}$ with $\|m-P(t^*)\|\leq\Delta$. As $m$ lies on $\overline{p\,q}$, we have that $\|p-q\|=\|p-m\|+\|m-q\|$. As $\|P(s)-P(t^*)\|\geq \lambda - 2\Delta$, we get that $\|p-m\|\geq \lambda-4\Delta$, and thus $\|m-q\|\leq 4\Delta$. And thus finally $\|P(t^*)-q\|\leq 5\Delta$, implying the claim.
\end{proof}

\begin{definition}
Let $P$ be a polygonal curve and $\lambda\geq 0$ and $\Delta\geq 0$. For $p\in\bR^d$ define the locus of edge endpoints of edges starting at $p$ which are close to subcurves of $P$ as the set
\[\cL_{\lambda,\Delta}(P,p)=\left\{q\in\bR^d\middle|\|p-q\|=\lambda \text{ and } \exists s,t \text{ with } 0\leq s\leq t\leq 1 \text{, }\df(P[s,t],\overline{p\,q})\leq\Delta\right\}.\]
\end{definition}

Similarly to Lemma~\ref{lem:psld} we can identify balls that cover the entirety of $\cL_{\lambda,\Delta}(P,p)$ for given $P,\lambda,\Delta$ and $p$. However, instead of a constant number of balls we need up to $k$ balls of constant radius to cover this set.

\begin{lemma}\label{lem:ppld}
Let $P\in\xx{d}{k}$ be a polygonal curve. Let $\lambda\geq 0$ and $\Delta\geq0$ be given. Then for every $p\in\bR^d$ there are $k$ points $p_1^*,\ldots,p_k^*\in\bR^d$ such that 
\[\cL_{\lambda,\Delta}(P,p)\subset\bigcup_{i=1}^k\ball_{5\Delta}(p_i^*).\]
\end{lemma}
\begin{proof}
The set $I=\{s\in[0,1]\mid P(s)\in\ball_{\Delta}(p)\}$ can be described as a disjoint union of at most $k$ closed intervals, as the complexity of $P$ is bounded by $k$. Assume that it is described by exactly $k$ such intervals, that is,  $I=\bigcup_{i=1}^k[l_i,r_i]$. 

\begin{figure}
    \centering
    \includegraphics[width=\textwidth]{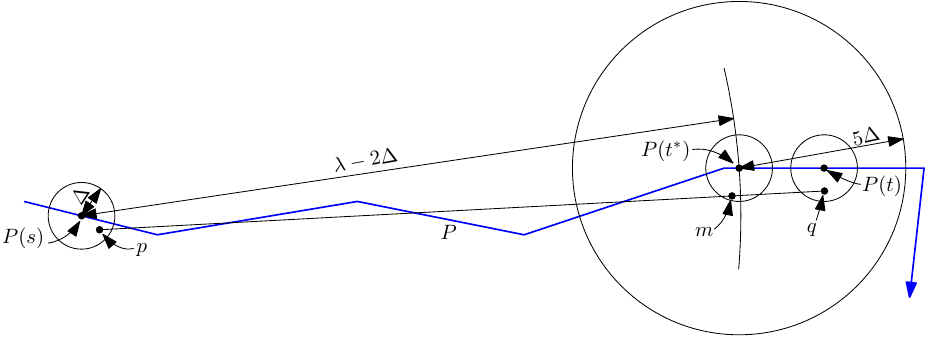}
    \caption{Illustration to the proof of Lemma~\ref{lem:psld}.}
    \label{fig:psld}
\end{figure}

It suffices to show that $\cL_{\lambda,\Delta}(P,p)\subset \bigcup_{i=1}^{k}L_{\lambda,\Delta}(P,l_i)$, as Lemma~\ref{lem:psld} then implies the claim. Assume that an arbitrary $q\in \cL_{\lambda,\Delta}(P,p)$ is given. Then by definition there are values $0 \leq s\leq t\leq 1$ such that $\df(P[s,t],\overline{p\,q})\leq\Delta$. This implies that $\|p-P(s)\|\leq\Delta$, and hence $s\in I$ and in turn $s\in[l_i,r_i]$ for some $1\leq i\leq k$. But then the subcurve $P[l_i,s]$ is contained in $\ball_\Delta(p)$, and thus $\df(P[l_i,t],\overline{p\,q})\leq\Delta$ implying that $q\in L_{\lambda,\Delta}(P,l_i)$ and thus the claim.
\end{proof}
\begin{corollary}\label{cor:plpdcover}
For every polygonal curve $P$ in $\bR^d$, $\Delta>0$, $\lambda>0$, $c>1$ and point $p\in\bR^d$ the set $N_{\Delta/c}\left(\cL_{\lambda,\left(1+c^{-1}\right)\Delta}(P,p)\right)$ can be covered by a set of balls of radius $\Delta/c$ centered at $O(k(10c+3)^d)$ points.
\end{corollary}
\begin{proof}
For any point $p\in\bR^d$, $\Delta>0$ and $c>1$ the sets $N_{\Delta/c}(D_{5\Delta}(p))$ and $D_{(5+c^{-1})\Delta}(p)$ coincide, so Lemma~\ref{lem:circlecover} and Lemma~\ref{lem:ppld} imply the claim.
\end{proof}

\begin{lemma}\label{lem:neighbourends}
Let $P$ be a polygonal curve. Let $\lambda\geq 0$, $\Delta\geq0$ and $c\geq 1$ be given. Then for every $p\in\bR^d$ and $p'\in\bR^d$ with $\|p-p'\|\leq\Delta/c$ we have that
\[\cL_{\lambda,\Delta}(P,p)\subset N_{\Delta/c}\left(\cL_{\lambda,\left(1+c^{-1}\right)\Delta}(P,p')\right).\]
\end{lemma}
\begin{proof}
Let $q\in \cL_{\lambda,\Delta}(P,p)$. This implies that there are values $0\leq s\leq t\leq 1$, such that $\df(P[s,t],\overline{p\,q})\leq\Delta$ and $\|p-q\|=\lambda$. Let $q'=q+(p'-p)$. Then, as $\|p-p'\|\leq\Delta/c$ and thus $\|q-q'\|\leq\Delta/c$, we get that $\df(\overline{p\,q},\overline{p'\,q'})\leq\Delta/c$ by Observation~\ref{obs:stabber}, and thus $\df(P[s,t],\overline{p'\,q'})\leq(1+c^{-1})\Delta$. Finally, $\|p'-q'\|=\|p-q\|=\lambda$ and thus the point $q'$ lies in $\cL_{\lambda,\left(1+c^{-1}\right)\Delta}(P,p')$, implying the claim. 
\end{proof}

We now prove a stronger version of Theorem~\ref{thm:main}, which allows us to analyze the doubling constant of $(\mu,\eps)$-curves in $\xx{d}{k}$.

\begin{lemma}\label{lem:main}
Let $k,\mu,d\in\bN$ and $\eps>0$. Let a $(\mu,\eps)$-curve $P$ in $\xx{d}{k}$ be given, as well as $\Delta>0$ and $c\geq 1$. There is a family of curves $\mathcal{C}_P\subset\xl{d}{k}{\mu\eps+\Delta/c}$ of size $O(k\mu(10c+3)^d)^k$, such that for any $(\mu,\eps)$-curve $Q$ with $\df(P,Q)\leq\Delta$ there is a $Q^*\in \mathcal{C}_P$ with $\df(Q,Q^*)\leq \Delta/c$. 
\end{lemma}

\begin{proof}





We construct the set $\mathcal{C}_P$ as follows. First, choose an element $(m_1,\ldots,m_{k-1})\in\{1,\ldots,\mu\}^{k-1}$. Next, choose one circle center of a cover of $\ball_\Delta(P(0))$ consisting of $O((2c+1)^d)$ many balls of radius $r/c$, which exists by Lemma~\ref{lem:circlecover}. Iteratively choose one point among the circle centers of a cover of $N_{\Delta/c}\left(\cL_{m_{i-1}\eps,(1+c^{-1})\Delta}(P,q_{i-1}^*)\right)$ of Corollary~\ref{cor:plpdcover}, consisting of $O(k(10c+3)^d)$ many balls of radius $r/c$ as the vertex $q_i^*$ of $Q^*$ for $i\leq k$. Then $Q^*\in \xl{d}{k}{\mu\eps+\Delta/c}$, as for any $i$ the fact that $q_{i}^*$ lies in $N_{\Delta/c}\left(\cL_{m_{i-1}\eps,(1+c^{-1})\Delta}(P,q_{i-1}^*)\right)$ implies that there is a point $q\in \cL_{m_{i-1}\eps,(1+c^{-1})\Delta}(P,q_{i-1}^*)$, with $\|q_{i-1}^*-q\|=m_{i-1}\eps$ and $\|q-q_i^*\|\leq\Delta/c$. Hence, $\|q_i^*-q_{i-1}^*\|\leq m_1\eps+\Delta/c\leq\mu\eps+\Delta/c$. To account for all the choices, we have that $|\mathcal{C}_P|=O(k\mu(10c+3)^d)^k$.

Let $Q$ be a given $(\mu,\eps)$-curve, with $\df(P,Q)\leq\Delta$. The curve $Q$ consists of $k-1$ edges and induces an ordered set $(m_1,\ldots,m_{k-1})\in\{0,\ldots,\mu\}^{k-1}$ representing the lengths of the edges in order. Let $q_1,\ldots,q_k$ be the vertices of $Q$. For all $1\leq i\leq k$ it holds that $q_i\in \cL_{m_{i-1}\eps,\Delta}(P,q_{i-1})$, by construction.

As $\df(P,Q)\leq\Delta$, the first vertex $q_1$ lies in $\ball_\Delta(P(0))$, and thus there is a point $q_1^*$ of the cover of $\ball_\Delta(P(0))$ consisting of balls of radius $r/c$, that lies at distance at most $\Delta/c$ to $q_1$. For every subsequent $q_i$, by Lemma~\ref{lem:neighbourends} and because $q_i\in \cL_{m_{i-1}\eps,\Delta}(P,q_{i-1})$, $q_i\in N_{\Delta/c}\left(\cL_{m_{i-1}\eps,(1+c^{-1})\Delta}(P,q_{i-1}^*)\right)$ and thus there is a point $q_i^*$ of the $\Delta/c$-cover of $N_{\Delta/c}\left(\cL_{m_{i-1}\eps,(1+c^{-1})\Delta}(P,q_{i-1}^*)\right)$ that is at distance at most $\Delta/c$ to $q_i$. This implies that there is an element $Q^*$ (defined by exactly this choice of points) in $\mathcal{C}_P$ that has distance $\ddf(Q,Q^*)\leq\Delta/c$ and thus, by Observation~\ref{obs:dfddf}, it holds that $\df(Q,Q^*)\leq\Delta/c$.
\end{proof}

We are now ready to prove Theorem~\ref{thm:main} as a straightforward application of Lemma~\ref{lem:main}.

\begin{proof}[Proof of Theorem~\ref{thm:main}]

Let $P$ be a $(\mu,\eps)$-curve in $\xx{d}{k}$ and a value $\Delta$ be given. By Lemma~\ref{lem:main}, there is a family $\mathcal{C}_P$ of curves of size $O(k\mu(43)^d)^k$ in $\xl{d}{k}{\mu\eps+\Delta/4}\subset\xx{d}{k}$, such that for any $(\mu,\eps)$-curve $Q$ with $\df(P,Q)\leq\Delta$ there is curve $Q^*$ in $\mathcal{C}_P$ with $\df(Q,Q^*)\leq\Delta/4$. For any $Q^*\in\mathcal{C}_P$ identify some $(\mu,\eps)$-curve $\widehat{Q^*}$ such that $\df(Q^*,\widehat{Q^*})\leq \Delta/4$. If no such element exists, ignore $Q^*$. Otherwise for any $(\mu,\eps)$-curve $Q$ with $\df(P,Q)\leq\Delta$ there is a curve $Q^*$ in $\mathcal{C}_P$ with $\df(Q,Q^*)\leq\Delta/4$, and thus by the triangle inequality there is a $(\mu,\eps)$-curve $\widehat{Q^*}$ with $\df(Q,\widehat{Q^*})\leq\Delta/2$, proving the bounded doubling dimension.
\end{proof}

\begin{corollary}
Let $k,\mu,d\in\bN$ and $\eps>0$. The doubling dimension of the space of $(\mu,\eps)$-curves in $\xx{d}{k}$ under the discrete Fréchet distance is bounded by $O(k(d+\log(k\mu)))$.
\end{corollary}
\begin{proof}
This is a consequence of the proof of Lemma~\ref{lem:main} and Theorem~\ref{thm:main}. By Observation~\ref{obs:dfddf}, for any two curves $P$ and $Q$ in $\xx{d}{k}$ it holds that $\df(P,Q)\leq\ddf(P,Q)$, and thus any $\Delta$-ball centered at a curve $P$ under the discrete Fréchet distance is contained in the $\Delta$-ball under the continuous Fréchet distance. In the proof of Lemma~\ref{lem:main} the $\Delta$-ball is covered by $\Delta/2$-balls under the discrete Fréchet distance, thus a proof similar to that of Theorem~\ref{thm:main} implies the claim.
\end{proof}

\subsection{Improvements for $c$-packed curves}

In this section, we add the assumption that the curves in the space of $(\mu,\eps)$-curves are $c$-packed, which leads to an improvement of the above bounds on the doubling dimension. 

\begin{definition}[$c$-packed]
A curve $P$ is said to be $c$-packed for some $c>0$, if for any ball of radius $r$, the length of $P$ inside the ball is at most $cr$. The set of $c$-packed $(\mu,\eps)$-curves in $\bX^{d,k}$ is defined as the intersection of the set of $(\mu,\eps)$-curves with the set of $c$-packed curves.
\end{definition}

\begin{lemma}\label{lem:cppld}
Let $P\in\xx{d}{k}$ be a polygonal $c$-packed curve with complexity at most $k$. Let $\lambda\geq 0$ and $\Delta\geq0$ be given. Then for every $p\in\bR^d$ there are $2c = O(c)$ points $p_1^*,\ldots p_{2c}^*\in\bR^d$ such that 
\[\cL_{\lambda,\Delta}(P,p)\subset\bigcup_{i=1}^{2c}\ball_{5\Delta}(p_i^*).\]
\end{lemma}
\begin{proof}
Assume $\lambda \geq 5\Delta$ as otherwise $\cL_{\lambda,\Delta}(P,p)\subset \ball_{5\Delta}(p)$ clearly holds, implying the claim.

For the sake of contradiction, assume that $\cL_{\lambda,\Delta}(P,p)$ cannot be covered by $2c$ balls of radius $5\Delta$. This implies, that there are at least $2c+1$ points $\{p_1,\ldots\}=:\mathcal{P} $ in $\cL_{\lambda,\Delta}(P,p)$ with a pairwise distance of at least $10\Delta$. For any $p_i\in\mathcal{P}$ we know that $\|p_i-p\|=\lambda$ and there are values $s_i<t_i$ such that $\df(P[s_i,t_i],\overline{p \, p_i})\leq \Delta$. For any two distinct $p_i,p_j\in\mathcal{P}$  the points $p_i,p_j$ and $p$ form an isosceles triangle with side lengths $\lambda,\lambda$ and $\|p_i-p_j\|\geq 10\Delta$. This implies that the distance to $p_j$ from any point $q$ along $\overline{p\,p_i}$ is at least $2\Delta$. This means that $t_j$ cannot lie in the interval $[s_i,t_i]$ as $\|p_j-P(t_j)\|\leq\Delta$. This in turn implies that all intervals $[s_1,t_1],\ldots$ are pairwise disjoint. We have thus identified $2c+1$ disjoint (in the domain) subcurves of $P$ with a total length of at least $(2c+1)(\lambda-2\Delta)$, contained in the ball $\ball_\lambda(p)$. However, since $\lambda\geq 5\Delta$, we have that $(2c+1)(\lambda-2\Delta)> c\lambda$, contradicting the fact that $P$ is $c$-packed. This in turn implies the claim. 
\end{proof}
\begin{corollary}\label{cor:main}
Let $k,\mu,d\in\bN$ and $c,\eps>0$.
The doubling constant of the space of $c$-packed $(\mu,\eps)$-curves in $\xx{d}{k}$ is bounded by $O(43^dc\mu)^k$ and thus its doubling dimension is bounded by $O(k(d+\log(c\mu)))$.
\end{corollary}
\begin{proof}
This follows from a straightforward modification of Theorem~\ref{thm:main} via Lemma~\ref{lem:cppld}.
\end{proof}

\section{Lower bounds for the doubling constant of $(\mu,\eps)$-curves}\label{sec:lowerbound}

In this section we want to show that the bound on the doubling dimension of $O(k(d+\log(k\mu)))$ is not too pessimistic. We begin with a straight-forward argument which implies a lower bound of $\Omega(d)$, before we discuss the lower bound of $\Omega(k\log \mu)$, resulting in a lower bound of $\Omega(d+k\log\mu)$.

The lower bound of $\Omega(d)$ follows trivially as the space of $(\mu,\eps)$-curves in $\bX^{d,1}$ consists of every singleton in $\bR^d$ and thus the doubling dimension of the $(\mu,\eps)$-curves in $\bX^{d,1}$ must be at least that of $\bR^d$. For spaces with curves of higher complexity any cover of a ball (with respect to $\df$) the bound follows similarly, by inspecting the starting points of the curves. 



\subsection{Lower bound of $\Omega(k\log \mu)$}\label{sec:lowerbound-mk}

In this section, we give a lower bound for the doubling dimension of both $c$-packed and non-$c$-packed $(\mu,\eps)$-curves in $\xx{d}{k}$ that shows the necessity of the factor $\mu^k$ in the bound of the doubling constant. The construction we give is an adaptation of a construction by Driemel et al. \cite{driemelClustering} that shows an unbounded doubling dimension for the space $(\xx{1}{3},\df)$.

\begin{lemma}\label{lem:lower}
Let $d=1$. Given $\mu>1$, $k\in\bN$ and $m\leq k/2$, we can construct a $(\mu,1)$-curve $C$, of complexity $k-2m$, and $\binom{(k-2m)(\mu)-1}{m}$ $(\mu,1)$-curves $G_i$, such that $d_F(C,G_i)\leq 1/2$ for all $i$. Additionally, for all $i\neq j$ there is no $(\mu,1)$-curve $X$ in $\xx{d}{k}$ such that $d_F(X,G_i)\leq 1/4$ and $d_F(X,G_j)\leq 1/4$. 
\end{lemma}


\begin{figure}
    \centering
    \includegraphics[width=\textwidth]{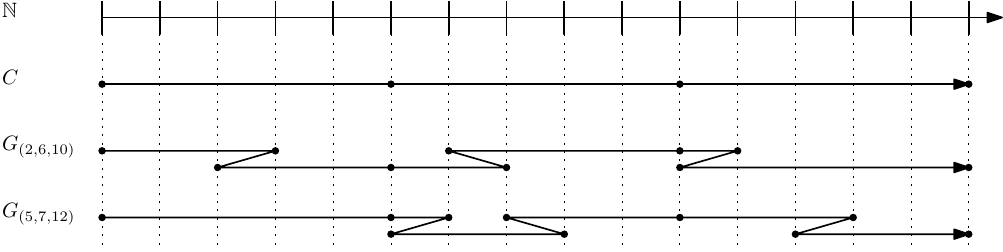}
    \caption{Illustration of the construction from Lemma~\ref{lem:lower} of two $(5,1)$-curves  $G_{(2,6,10)}$ and $G_{(5,7,12)}$ in $\xx{1}{9}$ that have Fréchet distance $1/2$ to the center curve $C$ at the top.}
    \label{fig:largedoubling}
\end{figure}

\begin{proof}

Define $C$ via the points $(c_1,\ldots,c_{k-2m+1})$, where $c_i = ((i-1)\cdot \mu)$ for $1\leq i\leq k-2m+1$. We now identify a large set of curves such that no two distinct elements of this set are at a distance at most $1/4$ to any $(\mu,1)$-curve in $\xx{1}{k}$. For this, choose an ordered subset $(n_1,\ldots,n_m)\subset\{0,\ldots,(k-2m)\mu-1\}$. 
Clearly, there are $\binom{(k-2m)\mu-1}{m}$ such choices. Based on the choice, we construct a curve $G_{(n_1,\ldots,n_m)}$ from $C$, first cutting $m+1$ pieces $C_0,\ldots,C_{m}$ from $C$, where $C_0$ goes from $0$ to $n_1+1$, $C_i$ from $n_i$ to $n_{i+1}+1$ for $1\leq i \leq m$, and $C_{m}$ from $n_m$ to $((k-2m)\mu)$.
For $1\leq i < m$, construct curves $T_i$ defined by two vertices $s_i=(n_i+1)$, $t_i=(n_i)$. 
Then, retrieve $G_{(n_1,\ldots,n_m)}$ via the concatenation $C_0* T_1 * C_1 * \ldots * T_m * C_m$. The set of curves constructed in this way forms the sought-after large set of curves. We show that $G_{(n_1,\ldots,n_m)}$ consists of $k$ edges. For this we show that every cut introduces $2$ new edges. If neither $n_i$ nor $n_i+1$ are vertices of $C$, then we are done, so assume that $n_i$ is a vertex of $C$. Then we introduce exactly two new edges: one from $(n_i)$ to $(n_i+1)$ and one from $(n_i+1)$ to $(n_i)$.
Similarly, if $n_i+1$ is a vertex of $C$, we also introduce two new edges, one from $(n_i+1)$ to $(n_i)$ and one from $(n_i)$ to $(n_i+1)$. 
Observe that $\df(C,G_{(n_1,\ldots,n_m)})=1/2$, as $C$ goes to the right, whereas $G_{(n_1,\ldots,n_m)}$ follows $C$ except in the introduced pieces $T_i$, where it goes left for a distance of $1$.

Let $X$ be a curve that has Fréchet distance at most $1/4$ to some curve $G=G_{(n_1,\ldots,n_m)}$ constructed above. Let $n_i\in (n_1,\ldots,n_m)$. Then the vertex $s_i=(n_i+1)$ and $t_i=(n_i)$ define the connecting piece $T_i$ between $C_{i-1}$ and $C_i$ of $G$. Now, $X$ has to first enter the interval $[n_i+3/4,n_i+5/4]$ and then $[n_i-1/4,n_1+1/4]$. As these two intervals are disjoint and $n_i < n_i+1$, the first interval lies to the right of the second interval. 
By construction, the vertex of $G$ before $s_i$ also lies to the left of $t_i$. Similarly, the vertex after $t_i$ lies to the right of $s_i$. Hence, $X$ has a vertex to the left of $t_i$, then a vertex in $[n_i+3/4,n_i+5/4]$, then a vertex inside $[n_i-1/4,n_i+1/4]$, and subsequently a vertex to the right of $s_i$. Thus, a curve $X$ also has to go to the left near $n_i+1$. 
Note that as every edge of $X$ has to have a length which is a multiple of $1$, any curve that does not go left near $n_i+1$ has to have a Fréchet distance of at least $1/2$ to $X$. 
This has to hold for all $n_i\in(n_1,\ldots,n_m)$. Since for any two distinct such constructed curves, there is a point where one travels to the left while the other does not, this then implies the claim as there are $\binom{(k-2m)(\mu)-1}{m}$ such constructed curves.
\end{proof}

\begin{corollary}\label{cor:lowermain}
For $d=1$ and given $\mu$ and $k$, the doubling dimension of the space of $(\mu,1)$-curves in $\xx{d}{k}$ is in $\Omega(k\log\mu)$.
\end{corollary}
\begin{proof}
Apply $m=k/3$ to Lemma~\ref{lem:lower}. This, together with the fact that \[\binom{(k/3)\mu-1}{k/3}\geq \frac{((k/3)\mu-1)^{(k/3)}}{(k/3)^{(k/3)}} = \Omega\left( \frac{((k/3)\mu)^{(k/3)}}{(k/3)^{(k/3)}}\right)=\Omega\left(\mu^{(k/3)}\right)\]
 holds implies the claim.
\end{proof}

\begin{theorem}\label{thm:lower}
For $d=1$ and given $\mu$, $k$ and $\eps>0$, the doubling dimension of the space of $(\mu,\eps)$-curves in $\xx{d}{k}$ is in $\Omega(d+k\log\mu)$.
\end{theorem}
\begin{proof}
This is a straight forward consequence of Corollary \ref{cor:lowermain} and the trivial $\Omega(d)$ lower bound.
\end{proof}

\begin{corollary}
For $d=1$, $c\geq 6$ and given $\mu$, $k$ and $\eps>0$, the doubling dimension of the space of $c$-packed $(\mu,\eps)$-curves in $\xx{d}{k}$ is in $\Omega(d+k\log\mu)$.
\end{corollary}
\begin{proof}
The constructed curves are clearly $6$-packed. Thus, Theorem \ref{thm:lower} implies the claim.
\end{proof}



\section{Approximate Nearest Neighbour}\label{sec:ANN}

Har-Peled et al.~\cite{harpeledFast} showed that $(1+\eps)$-ANN can be solved in metric spaces of bounded doubling dimension.

\begin{theorem}[\cite{harpeledFast}]\label{thm:harpeled}
Given a set $S$ of $n$ points in a metric space $\cM$ of bounded doubling dimension $\nu$, one can construct a data structure for answering $(1+\eps)$-approximate nearest neighbour queries. The query time is $2^{O(\nu)}\log n + \eps^{-O(\nu)}$, the expected preprocessing time is $2^{O(\nu)}n\log n$ and the space used is $2^{O(\nu)}n$.
\end{theorem}

A careful reading reveals an important specification for our purposes, namely, that the doubling dimension is that of the $n$-point metric space defined by $S$ induced by the metric space $\cM$ and not of the ambient metric space $\mathcal{M}$. Note that by Lemma \ref{lem:subspacedoubling} the doubling dimension of the metric space induced on the subset is at most twice the doubling dimension of the ambient space. For an example where the doubling dimension increases, refer to Figure \ref{fig:subspacedoublingexample}. 

\begin{figure}
    \centering
    \includegraphics[width=\textwidth]{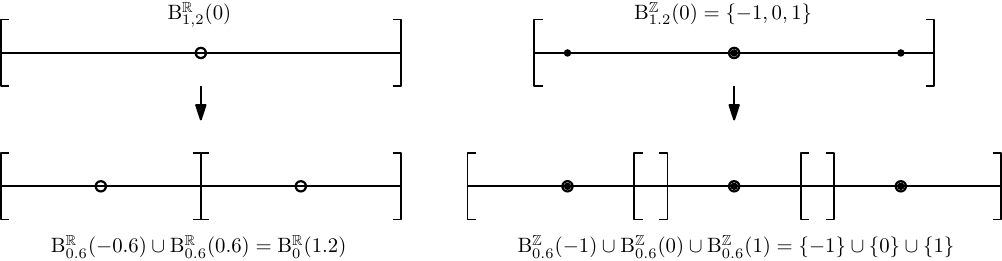}
    \caption{Example of a subset $\bZ$ of the metric space $\bR$ whose doubling dimension is larger than that of its ambient space. The disk centers are marked by circles.}
    \label{fig:subspacedoublingexample}
\end{figure}

\begin{lemma}\label{lem:subspacedoubling}
Let $(\mathcal{M},\dd_\cM)$ be a metric space, and let $S$ be some subset of $\cM$. Then the doubling dimension of $(S,\dd_\cM)$ is at most twice the doubling dimension of $(\cM,\dd_\cM)$.
\end{lemma}
\begin{proof}
Let $\nu$ be the doubling dimension of $\cM$.
Let $s\in S$ and $r>0$ be given and let $\ball^S_r(s)\subset S$ be the ball in $S$, of radius $r$ and centered at $s$, i.e., $\ball^S_r(s) = \ball^\cM_r(s)\cap S$. But $\ball_r^\cM(s)$ can be covered by $(2^\nu)^2$ balls of radius $r/4$. For any such ball check if the intersection with $S$ is nonempty. If this is the case, pick some element from this intersection, and center a ball in $S$ of radius $r/2$ around it. Clearly, any such larger ball contains the intersection of the smaller ball with $S$. Therefore, $\ball^S_r(s)$ is contained in the union of at most $(2^\nu)^2$ balls of radius $r/2$ in $S$, which implies the claim, as $\log_2((2^\nu)^2)=2\nu$.
\end{proof}

By Theorem~\ref{thm:main}, we know that the doubling dimension of the space of $(\mu,\eps)$-curves in $\xx{d}{k}$ is bounded. We further know that for any $\eps>0$ we can map any curve of $\xl{d}{k}{\Lambda}$ into the space of $(\lceil\Lambda/\eps\rceil+1,\eps)$-curves in  $\xx{d}{k}$ with a distortion of at most $\eps/2$, by Lemma~\ref{lem:simplification}. Hence, Theorem~\ref{thm:harpeled} together with Lemma~\ref{lem:subspacedoubling} imply Theorem~\ref{lem:approxresult}, a central piece to constructing a data structure solving the $(1+\eps)$-ANN problem for polygonal curves under the Fréchet distance.


\lemapproxresult*
\begin{proof}
Define $\hat{\eps}=\eps'/2$.
Let $\mu = \lceil\Lambda/\hat{\eps}\rceil+1=\Theta(1+\Lambda/\eps')$.
We begin by simplifying every polygonal curve $s\in S$ via Lemma~\ref{lem:simplification}, resulting in a set $S'$ of $({\mu},{\hat{\eps}})$-curves. This takes $O\left(\log(\mu) n k\right)$ time, which is in $O\left(2^\nu n\right)$. As $S'$ lies in the space of $({\mu},{\hat{\eps}})$-curves, the doubling dimension of the set $S'$ with the Fréchet distance is bounded by $\nu=O(k(d+\log(k(1+\Lambda/\eps'))))$ via Theorem~\ref{thm:main} and Lemma~\ref{lem:subspacedoubling}. Note that for every $s\in S$ and its simplification $s'\in S'$ it holds that $\df(s,s')\leq\hat{\eps}/2$. We apply Theorem~\ref{thm:harpeled} to the set $S'$ and $\eps$. Note that Theorem~\ref{thm:harpeled} assumes that the distance between any two points in the metric space of $(\mu,\eps)$-curves can be computed in $O(1)$ time. However, the computation of the continuous Fréchet distance takes polynomial time in $k$. On the other hand, both $2^{k\log k}$ and $\eps^{-k\log k}$ dominate $\mathrm{poly}(k)$ for $\eps<1$. Thus the running time is indeed as claimed.
We then query the data structure with $q$, returning an element $\widehat{s'}$ such that for every $s'\in S'$ it holds that $\df(q,\widehat{s'})\leq(1+\eps)\df(q,s')$. Lastly, the element of $S$ returned by the data structure will be the element $\widehat{s}\in S$ which corresponds to $\widehat{s'}$. We then get for every $s\in S$ that 
\begin{align*}
\df(q,\widehat{s})&\leq\df(q,\widehat{s'})+\hat{\eps}/2\leq(1+\eps)\df(q,s')+\hat{\eps}/2\leq(1+\eps)(\df(q,s)+\hat{\eps}/2)+\hat{\eps}/2 \\
&\leq (1+\eps)\df(q,s)+\hat{\eps}+\eps\hat{\eps}/2=(1+\eps)\df(q,s)+\hat{\eps}(1+\eps/2)\\
&\leq(1+\eps)\df(q,s)+\eps'.\qedhere
\end{align*}
\end{proof}

To get rid of the additive error, we want to set $\eps'$ to $O(\min_{s\neq s'\in S}\df(s,s'))$. As $\eps'$ impacts the running time, we reformulate the running time in terms of the spread.

Note that this measure is scale- and translation-independent, as translating all elements in $S$ by the same offset changes neither their pairwise Fréchet distances nor the length of the edges of the curves. Similarly scaling the elements of $S$ scales both the pairwise Fréchet distances as well as the length of the edges of the curves, thus not changing the \textgirth{}.

\begin{lemma}\label{lem:girthspread}
Given a set of curves $S\in\xx{d}{k}$, then 
\[ \mathcal{G}(S)^{-1}=O(\Phi(S))\]
where $\Phi(S)$ denotes the spread of the set of vertices and edges of curves in $S$.
\end{lemma}
\begin{proof}
Observe that no two curves in $S$ have a Fréchet distance of $0$. Note that the Fréchet distance between two curves $P$ and $Q$ is approximated up to a constant by the Euclidean distance of either one vertex of $P$ and one vertex of $Q$ or the Euclidean distance of a vertex and an edge (one of $P$ and one of $Q$) as sets. Thus \[\min_{s\neq s'\in S}\df(s,s')= \Omega\left(\min_{\substack{o,o'\in V(S)\cup E(S)\\\mathrm{d}(o,o')> 0}}\mathrm{d}(o,o')\right)\] where $V(S)$ denotes the set of vertices and $E(S)$ denotes the set of edges defining the curves in $S$. Further note that as the length of any edge of a curve $s$ in $S$ is defined as the distance of two vertices defining $s$. As such, observe that
\[\max_{s\in S, e\in E(s)}\|e\|\leq\max_{\substack{p,q\in V(S)}}\mathrm{d}(p,q)\leq \max_{\substack{o,o'\in V(S)\cup E(S)}}\mathrm{d}(o,o').\]
Thus the claim follows.
\end{proof}

\thmgirthmain*
\begin{proof}
Let $\eps'=\eps/4$ and $\eps''=\eps'\left(\min_{s\neq s'\in S}\df(s,s')\right)$. Let $E(S)$ be the set of edges of curves in $S$ and let further $\Lambda=\max_{e\in E(S)}\|e\|$, thus clearly $S\subset\xl{d}{k}{\Lambda}$. We then apply Theorem~\ref{lem:approxresult} with $\eps'$ and $\eps''$ resulting in the described data structure. Let $q\in\xx{d}{k}$ be given. Let $s^*$ be the element in $S$ minimizing the distance to $q$. Querying the data structure with $q$ results in an element $\widehat{s}$ with the the property that
\[\df(q,\widehat{s})\leq(1+\eps')\df(q,s^*) + \eps'\left(\min_{s\neq s'\in S}\df(s,s')\right).\]
We now look at two different cases. First assume $(2+\eps')\df(q,s^*) < (1-\eps')\min_{s\neq s'\in S}\df(s,s')$. But then for every $s'\neq s^*$ we know that
\begin{align*}
    \df(q,s')&\geq \df(s^*,s') - \df(q,s^*)\geq \min_{s\neq s'\in S}\df(s,s') - \df(q,s^*)\\
    &>(1-\eps')\left(\min_{s\neq s'\in S}\df(s,s')\right) + \eps'\left(\min_{s\neq s'\in S}\df(s,s')\right)- \df(q,s^*)\\
    &= (1+\eps')\df(q,s^*) + \eps'\left(\min_{s\neq s'\in S}\df(s,s')\right)
\end{align*}
and thus $\widehat{s} = s^*$, implying $\df(q,\widehat{s})\leq (1+\eps)\df(q,s^*)$.

Now assume $(2+\eps')\df(q,s^*) \geq (1-\eps')\min_{s\neq s'\in S}\df(s,s')$. Then we know that 
\begin{align*}
    \df(q,\widehat{s}) &\leq (1+\eps')\df(q,s^*) + \eps'\left(\min_{s\neq s'\in S}\df(s,s')\right)\\
    &\leq (1+\eps')\df(q,s^*) + \eps'\left(\frac{2+\eps'}{1-\eps'}\right)\df(q,s^*).
\end{align*}
Now since $\eps\leq 1$, we know that $\eps'\leq 1/4$ and thus $\frac{2+\eps'}{1-\eps'}\leq 3$. This then concludes the case-distinction, as
\[\df(q,\widehat{s})\leq (1+\eps')\df(q,s^*) + \eps'\left(\frac{2+\eps'}{1-\eps'}\right)\df(q,s^*) \leq (1+4\eps')\df(q,s^*) = (1+\eps)\df(q,s^*).\]

Regarding the running time, observe that 
$\Lambda/\min_{s\neq s'\in S}\df(s,s')=\mathcal{G}(S)^{-1}$. Hence, as $\eps' = \Theta(\eps)$ and $\eps''=\Theta(\eps\left(\min_{s\neq s'\in S}\df(s,s')\right))$, the preprocessing time, query time and space is as claimed.
\end{proof}

\corgirthmain*
\begin{proof}
This is a direct consequence of Theorem~\ref{thm:girth_main} and Lemma~\ref{lem:girthspread} together with the fact that the spread of any collection of sets is at least $1$, implying $(1+\Phi(S)\eps^{-1})=O(\Phi(S)\eps^{-1})$, since both $\Phi(S)\geq 1$ and $\eps \leq 1$.
\end{proof}

\section{Generalization to arbitrary metric spaces}\label{sec:generalization}

In this section, we observe that our techniques directly generalize to other metric spaces. 
In particular, 
we can solve $(1+\eps)$-ANN in any space with unbounded doubling dimension as long as a subspace 
with bounded doubling dimension that is 
close (under the Gromov-Hausdorff distance) exists. This approach is closely related to that of Sheehy and Seth \cite{sheehyNearlydsopd} in that they extended Clarkson's algorithm for finding a $\lambda$-net in some metric space $X$ if there is another space of bounded doubling dimension with Gromov-Hausdorff distance at most $\lambda/3$ to $X$.



\begin{definition}[tractably nearly-doubling space]
Let $\mathcal{M}$ be a metric space. We say $\mathcal{M}$ is tractably nearly-doubling if there are functions $\nu_\mathcal{M}:\bR\rightarrow\bR$ and $\rho_\cM:\cM\times\bR\rightarrow\bR$ such that for every $\eps>0$ there is some subspace $\mathcal{M}_\eps\subset\mathcal{M}$ with doubling dimension at most $\nu_\mathcal{M}(\eps)$ and projection $\pi_\eps:\mathcal{M}\rightarrow\mathcal{M}_\eps$. Furthermore we require $\dd_\cM(s,\pi_\eps(s))\leq\eps$ for any $s\in\cM$ and the element $\pi_\eps(s)\in\cM_\eps$ can be computed in $\rho_\cM(s,\eps)$ time.
\end{definition}

Observe that for a tractably nearly-doubling space $\mathcal{M}$ the Gromov-Hausdorff distance $\mathrm{d}_{GH}(\mathcal{M},\mathcal{M}_\eps)$ is at most $\eps$. Note that by Lemma~\ref{lem:simplification} and Theorem~\ref{thm:main}, the space $\xl{d}{k}{\Lambda}$ is tractably nearly-doubling for any $d,k\in\bN$ and $\Lambda>0$. The proof of the following lemma and theorem conceptually are the same proof as for Theorem~\ref{lem:approxresult} and Theorem \ref{thm:girth_main}.

\begin{lemma}\label{lem:generalized}
Let $\mathcal{M}$ be a metric space, and let $S\subset \mathcal{M}$ be a set of $n$ points and $\eps>0$. Let $\mathcal{M}'\subset \mathcal{M}$ be some subspace of doubling dimension $\nu$, and let $S'\subset\mathcal{M}'$ be another set with a surjection $\pi:S\rightarrow S'$ with a running time of $O(1)$. 
Then one can compute a data structure which for given $q\in\mathcal{M}$ outputs $\hat{s}\in S$, such that for every $s\in S$
\[\mathrm{d}_\mathcal{M}(q,\hat{s})\leq(1+\eps)\mathrm{d}_\mathcal{M}+(2+\eps)\left(\max_{s\in S}\dd_\cM(s,\pi_\eps(s))\right).\]
The expected preprocessing time is $2^{O(\nu)}+\eps^{-O(\nu)}$ and the space used is $2^{O(d)}n\log n$, with the query time being $2^{O(\nu)}\log n + \eps^{-O(\nu)}$, where $\nu$ is the doubling dimension of $\mathcal{M}'$.
\end{lemma}

\begin{proof}
This proof is parallel to the proof of Theorem~\ref{lem:approxresult}. We first store the inverse map $\pi^{-1}:S'\rightarrow S$ in $O(n)$ time. We then compute the data structure from Theorem~\ref{thm:harpeled} on the set $S'$. This takes $2^{O(\nu)}+\eps^{-O(\nu)}$ expected preprocessing time and $2^{O(\nu)}n\log n$ space. Further the query time for this data structure is $2^{O(\nu)}\log n + \eps^{-O(\nu)}$, as $\nu$ is the doubling dimension of $\mathcal{M}'$.

Then for any $q\in\mathcal{M}$ the data structure outputs a point $\hat{s}'\in S'$. We then output $\hat{s}=\pi^{-1}(\hat{s}')$. Then overall for every point $s'\in S'$ it holds that $\mathrm{d}_{\mathcal{M}}(q,\hat{s}')\leq(1+\eps)\mathrm{d}_\mathcal{M}(q,s')$. Then for any $s'\in S'$ it holds that
\[\mathrm{d}_\mathcal{M}(q,\hat{s}) \leq \mathrm{d}_\mathcal{M}(q,\hat{s}') + \mathrm{d}_\mathcal{M}(\hat{s},\hat{s}')\leq (1+\eps)\mathrm{d}_\mathcal{M}(q,s') + \max_{s\in S}\dd_\cM(s,\pi_\eps(s))\]
And thus for any $s\in S$ it holds that 
\begin{align*}
    \mathrm{d}_\mathcal{M}(q,\hat{s}) &\leq(1+\eps)(\mathrm{d}_\mathcal{M}(q,s)+\mathrm{d}_\mathcal{M}(s,s')) + \max_{s\in S}\dd_\cM(s,\pi_\eps(s))\\
    &\leq (1+\eps)\mathrm{d}_\mathcal{M}(q,s) + (2+\eps)\left(\max_{s\in S}\dd_\cM(s,\pi_\eps(s))\right),
\end{align*}
implying the claim.
\end{proof}

\begin{theorem}\label{thm:generalized}
Let $\mathcal{M}$ be a metric space that is tractably nearly-doubling. Let $S\subset\mathcal{M}$ be a set of $n$ points. Then for every $\eps>0$ one can construct a data structure which for any given $q\in\mathcal{M}$ returns a point $\hat{s}\in S$ such that for every $s\in S$ it holds that
\[\mathrm{d}_\mathcal{M}(q,\hat{s})\leq(1+\eps)\mathrm{d}_\mathcal{M}(q,s).\]
The expected preprocessing time is given by $\sum_{s\in S}\rho_\cM(s,e)+2^{O(\nu_{\cM}(e))}+\eps^{-O(\nu_{\cM}(e))}$ and the space used is in $2^{O(\nu_{\cM}(e))}n\log n$. The query time is given by $2^{O(\nu_{\cM}(e))}\log n + \eps^{-O(\nu_{\cM}(e))}$, where $e=(\eps/4)\min_{s\neq s'\in S}\mathrm{d}_\mathcal{M}(s,s')$.
\end{theorem}

\begin{proof}

This proof follows the proof of Theorem~\ref{thm:girth_main}. Just like before, it is a consequence from Lemma~\ref{lem:generalized}, by selecting specific parameters.

First let $\eps'=\eps/4$ and $\eps''=\eps\min_{s\neq s'\in S}\mathrm{d}_\mathcal{M}(s,s')$. We then apply Lemma~\ref{lem:generalized} with the subspace $\mathcal{M}_{\eps''}$, the set $\pi_{\eps''}(S)\subset\mathcal{M}_{\eps''}\subset\mathcal{M}$ and the map $\pi_{\eps''}$, where the doubling dimension of $\mathcal{M}_{\eps''}$ is given by $\nu_{\cM}(\eps'')$, as $\mathcal{M}$ is tractably nearly-doubling. Thus $\mathrm{d}_H(S,\pi_{\eps''}(S))\leq\eps''$. The expected preprocessing time is given by $\sum_{s\in S}\rho_\cM(s,\eps) + 2^{O(\nu_{\cM}(\eps''))}+\eps^{-O(\nu_{\cM}(\eps''))}$ and the space used is in $2^{O(\nu_{\cM}(\eps''))}n\log n$. The query time is given by $2^{O(\nu_{\cM}(\eps''))}\log n + \eps^{-O(\nu_{\cM}(\eps''))}$.

Let $q\in\mathcal{M}$ be given. Let $\hat{s}$ be the output of the data structure, and let $s^*$ be the point in $S$ minimizing the distance to $q$. We again start, by assuming $(2+\eps')\mathrm{d}_\mathcal{M}(q,s^*)\leq(1+\eps')\min_{s\neq s'\in S}\mathrm{d}_\mathcal{M}(s,s')$. But then for every $s'\neq s^*$ we know that

\begin{align*}
    \mathrm{d}_\mathcal{M}&\geq\mathrm{d}_\mathcal{M}(s^*,s') - \mathrm{d}_\mathcal{M}(q,s^*)\geq \min_{s\neq s'\in S}\mathrm{d}_\mathcal{M}(s,s')-\mathrm{d}_\mathcal{M}(q,s^*)\\
    &>(1-\eps')(\min_{s\neq s'\in S}\mathrm{d}_\mathcal{M}(s,s')) + \eps'(\min_{s\neq s' \in S}) - \mathrm{d}_\mathcal{M}(q,s^*)\\
    &=(1+\eps')\mathrm{d}_\mathcal{M}(q,s^*) + \eps'(\min_{s\neq s' \in S}\mathrm{d}_\mathcal{M}(s,s'))
\end{align*}
and thus $\hat{s} =s^*$, implying $\mathrm{d}_\mathcal{M}(q,\hat{s})\leq (1+\eps)\mathrm{d}_\mathcal{M}(q,s^*)$.

Now assume $(2+\eps')\mathrm{d}_\mathcal{M}(q,s^*) \geq (1-\eps')\min_{s\neq s'\in S}\mathrm{d}_\mathcal{M}(s,s')$. Then we know that 
\begin{align*}
    \mathrm{d}_\mathcal{M}(q,\widehat{s}) &\leq (1+\eps')\mathrm{d}_\mathcal{M}(q,s^*) + \eps'\left(\min_{s\neq s'\in S}\mathrm{d}_\mathcal{M}(s,s')\right)\\
    &\leq (1+\eps')\mathrm{d}_\mathcal{M}(q,s^*) + \eps'\left(\frac{2+\eps'}{1-\eps'}\right)\mathrm{d}_\mathcal{M}(q,s^*).
\end{align*}
Now since $\eps\leq 1$, we know that $\eps'\leq 1/4$ and thus $\frac{2+\eps'}{1-\eps'}\leq 3$. This then concludes the proof, as
\begin{align*}
    \mathrm{d}_\mathcal{M}(q,\widehat{s})&\leq (1+\eps')\mathrm{d}_\mathcal{M}(q,s^*) + \eps'\left(\frac{2+\eps'}{1-\eps'}\right)\mathrm{d}_\mathcal{M}(q,s^*) \\
    &\leq (1+4\eps')\mathrm{d}_\mathcal{M}(q,s^*) = (1+\eps)\mathrm{d}_\mathcal{M}(q,s^*).\qedhere
\end{align*}
\end{proof}
\section{Conclusion}

In this work we established that despite the unbounded doubling dimension of the metric space of the Fr\'echet distance of curves $\bX^{d,k}$, there are spaces which are arbitrarily close to $\bX^{d,k}$ which all have bounded doubling dimension, depending on the distance to $\bX^{d,k}$. We then constructed an approximate nearest neighbour data structure for $\bX^{d,k}$ by answering approximate nearest neighbour queries in these arbitrarily close spaces via well-established data structures which can be constructed in spaces of bounded doubling dimension. As a special case we considered the setting of $c$-packed curves to which our approach extends, resulting in an improved running time in this special setting.

We further gave a constructive argument as to why the doubling dimension of the space of $(\mu,\eps)$-curves in $\xl{d}{k}{\Lambda}$ is large. The gap between the upper bound on the doubling dimension and the construction is quite small --- especially in the case of $c$-packed curves for constant $c$ --- but any further improvements to this gap would be interesting to see. We do not believe that either of the given bounds are necessarily tight.

Intuitively, it seems reasonable to assume that the 'dimension' of the problem should be at least $kd$ for $d \in O(\log nm)$ and $k=m$.
It is well-known~\cite{indyk20178} that the $L_{\infty}$-metric in $\mathbb{R}^d$ embeds isometrically into a Fr\'echet metric space of one-dimensional curves of complexity $k=3d$. 
Recently, Rohde and Psarros showed that random projections can be used to  obtain dimensionality reductions for the Fr\'echet metric~\cite{Psarros2022RandomProjections} when $d \in \Omega(\log(nm))$.

An important future research direction the reduction of the dependence on the \textgirth{} (or spread of the underlying set of vertices and edges) in the running time. In our case the dependence on the spread is a result of turning the additive error of the Gromov-Hausdorff distance between $\bX^{d,k}$ and the space of $(\mu,\eps)$-curves into a multiplicative error. With our approach this seems to be inevitable.

It would further be interesting to see what other results from spaces of bounded doubling dimension can be extended to spaces (not restricted to that of polygonal curves) of unbounded doubling dimension in this way. 

\bibliography{mybiburls}{}
\bibliographystyle{plainurl}

\newpage

\appendix 









\end{document}